\documentclass[reqno]{amsart}
\usepackage{amsmath}
\usepackage{amsfonts}
\usepackage{amssymb}
\usepackage{graphicx}
\usepackage{amsthm}

\numberwithin{equation}{section}

\theoremstyle{plain}
\newtheorem{theorem}{Theorem}[section]
\newtheorem{proposition}[theorem]{Proposition}
\newtheorem{lemma}[theorem]{Lemma}
\newtheorem{corollary}[theorem]{Corollary}

\newtheorem{example}[theorem]{Example}

\renewcommand{\P}{\mathbb{P}}
\newcommand{\E}{\mathbb{E}}
\newcommand{\R}{\mathbb{R}}
\newcommand{\F}{\mathcal{F}}

\title{High Frequency Market Making}
\date{\today}
\author{Ren\'e Carmona}
\author{Kevin Webster}
\begin{document}

\begin{abstract}
Since they were authorized by the U.S. Security and Exchange Commission in 1998, electronic exchanges have boomed, and by 2010 high frequency trading accounted for over 70\% of equity trades in the US. Such markets are thought to increase liquidity because of the presence of market makers, who are willing to trade as counterparties at any time, in exchange for a fee, the bid-ask spread. In this paper, we propose an equilibrium model showing how such market makers provide liquidity. The model relies on a codebook for client trades, the implied alpha. After solving the individual clients optimization problems and identifying their implied alphas, we frame the market maker stochastic optimization problem as a stochastic control problem with an infinite dimensional control variable. Assuming either identical time horizons for all the clients, or a stochastic partial differential equation model for their beliefs, we solve the market maker problem and derive tractable formulas for the optimal strategy and the resulting limit-order book dynamics.
\end{abstract}

\maketitle

\section{Introduction}

Electronic exchanges play an increasingly important role in financial markets and market microstructure became the key to understanding them. Here market microstructure is understood as \emph{the study of the trading mechanisms used for financial securities} \cite{Hasbrouck}. High frequency trading strategies depend very strongly on these mechanisms which in turn, vary from market to market. Three main themes were proposed to unite the common features of most of these markets \cite{Hasbrouck}. 
\begin{enumerate}
\item The first theme is the limit-order book, where agents can post trading intentions at prices above (ask) or below (bid) the mid-price, which is regarded as the current \emph{fair price} for the security. These posts, called \emph{limit-orders}, can then lead to executions if they are paired with a matching order. In this case, the initial agent, commonly known as the liquidity provider, gains a small amount of money, having sold at the ask or bought at the bid. It is said in that case that the liquidity taker has paid a \emph{liquidity fee} for the right to trade immediately.
\item The second theme is adverse selection due to asymmetric sources of information. It is based on the principle that liquidity takers choose which limit-orders are actually executed, and when. Moreover, since liquidity providers publicly announce their intentions by posting their limit-orders, and these two market rules create an \emph{information bias} in favor of liquidity takers that compensates for the \emph{liquidity fee}.
\item The last theme is statistical predictions. All agents on electronic markets attempt to predict prices on a short time scale. Such predictions are possible both because of the existence of private information and the actual market mechanisms that create short-range price correlations. While statistical predictions are used by all traders on electronic markets, the objectives of these agents vary widely. Some make their profits from these predictions, while others just reduce their transaction costs by \emph{trading smartly}.
\end{enumerate}

After the liquidity crisis of 2008, liquidity became the central object of interest for market microstructure. Liquidity is rarely  defined precisely, although intuitively, it should quantify how difficult it is  to engage in a trade. Commonly accepted measures of liquidity are the bid-ask spread\footnote{The bid-ask spread is defined as the price difference between the lowest ask and the highest bid in the limit-order book.} and the volume present in a limit-order book. We tie these two quantities together by modeling what we will call the liquidity cost curve, or simply cost curve. We define this to be the total fee paid by the liquidity taker as a function of the volume asked for. If the liquidity taker only executes limit-orders at the best ask or bid, the fee will be equal to the bid-ask spread times the volume. Some clients, however, might wish to go beyond the best bid or ask and take more volume from the market. The marginal fee then becomes a function of the volume.

Trade volumes are difficult to model, in part because they are so dependent on liquidity takers' decisions and their execution strategies. Empirically, trade volumes present very heavy tails and are strongly autocorrelated (\cite{Bouchaud3}): there are frequent outliers, buys tend to follow buys, and sells tend to follow sells, which is what makes them hard to evaluate statistically. Limit-orders, on the other hand, are mean-reverting and have been widely studied in the literature (see, for instance, \cite{Bouchaud3,Cont,Farmer}). The interplay between trend-following liquidity takers and mean-reverting liquidity providers is what makes the market reach the critical state of diffusion, as seen in \cite{Bouchaud}. From a modeling stand point, all these attributes are what makes a direct statistical description of trade volumes so difficult, which is why the introduction of a codebook is desirable.

\subsection{Market making}

Market makers are a special class of liquidity providers. They essentially act as a scaled down version of the market itself, always providing limit-orders on both sides of the mid-price. How these limit-orders are placed in terms of volume, distribution, and distance from the mid-price, determines the pricing strategy. We can therefore define the liquidity curve associated with a market maker's pricing strategy. The agents trading with a market maker are considered as his \emph{clients}. Market making is a purely passive strategy that essentially corresponds to the service of providing liquidity to the market against a fee. This strategy makes money as long as the pricing accurately anticipates future price variations. The first step towards that objective is to find a model for client volumes that is consistent with price dynamics.

There are two schools of thoughts on market making models. The first focuses on inventory risk. There, the market maker has a preferred inventory position and prices according to his risk aversion to diverging inventory levels. Some of the first models of this type are \cite{Garman,Amihud}, which can be found in \cite{Hasbrouck}. The second, initiated by \cite{Kyle}, focuses on adverse selection, usually distinguishing between informed and noise trades. This paper belongs to the second line of thought.

According to our stylized definition of a market maker, the latter does not possess a view on the market. Clients, on the other hand, have views on the market and this leads to trading needs. When this view is short term, then the client has a statistical edge on the rest of the market and tries to push that edge to make a short term profit. But even when the view is long term, the client will still attempt to minimize transaction costs by predicting prices on a short time scale. Long-term strategies and liquidity constraints can therefore be modeled as noise around optimal short term execution strategies. The intuition behind this is that long term constraints, while the main motives behind most trades, will not necessarily dictate their execution. This is the case because there are increasingly more layers between the entity that formulates the long term strategy (for example investment banks or hedge funds) and the one that actually executes it (for example brokers or execution engines). But only the latter actually has a short term view on the market, which is what the market maker is interested in. Optimal executions based on short term beliefs and associated trade volumes form therefore the cornerstone of the upcoming analysis and will link prices to trades.

\subsection{Alpha}

Alpha is the term often used to signify that a client has a directional view on the market. In an adverse selection model the aim of any market maker is to discover this view and hedge against it. Most papers (see for example \cite{Almgren,Bouchaud,Bouchaud2,Hasbrouck}) introduce a notion of market impact or response function,  to study the relationship between trades and prices. Trade size is often  -- though not always -- included, and a lot of work has been done to fit various curves to the response function. A large number of papers on execution strategies (\cite{Alfonsi,Almgren,Almgren2,Wang}) rely heavily on this notion of market impact. So do most practitioners. Essentially, alpha and market impact are the same quantity, even though the stories told to justify these concepts are quite different. In the market impact literature, the premise is that a trade \emph{causes} a price movement, while alpha is viewed as an attempt at \emph{predicting} the movement. Response function is a more neutral term that allows for both interpretations.

In this paper, we present a simple model for client decisions that proves, under very general assumptions, a systematic relationship between trade sizes and short term expected price variations. It states that marginal costs should equal expected price variations. This is an intuitive result given that price variations are the marginal gains of a risk-neutral client. In this model, each client has a short term view on the market and trades optimally according to this view: his optimal execution strategy depends upon his belief. The theoretical notion of \emph{implied alpha} appears quite naturally from the stochastic optimal control problem that defines the client's model.

On the market maker side, our model tries to capture the fact that market makers do not act on a personal view of the market. The market maker uses the notion of client implied alpha to infer an approximate price process from his clients' behaviors, essentially aggregating the views of his clients to form a probability distribution on the price. The optimal order book strategy then replicates this distribution, with a corrective term that takes into account the profitability of trades, that is, the trade-off between spread and volume.

\subsection{Results}

The thrust of this paper is to propose a framework in which a market making strategy appears endogenously. This framework is based on
\begin{enumerate}
\item a result to imply a market view from client trades under a simple yet robust model of client behavior;
\item a procedure for the market maker to infer his own view on the market from that of his clients;
\item a profitability function that measures how profitable a posted volume on the order book is likely to be;
\item a penalizing term taking into account possible feedback effects of the market maker's decision.
\end{enumerate}
The last part will be motivated by the idea that, just as the market maker implies information from his clients' trades, clients can infer information from his order book. Our hypotheses are chosen, and our results are derived, in order to make the market making problem tractable. Once the theoretical framework is put in place, the market maker optimal control problem is formulated and solved, leading to an explicit market making strategy.

\vskip 2pt
The story told by the model is closest to that of the informed trade literature \cite{Kyle}, although one major difference is that there is no clear cut distinction between informed and noise traders. Our conservative market maker assumes that all of his clients' implied alphas carry information. Moreover, the existence of at least one informed trader underpins the estimation formula used by the market maker.
The paper also relates strongly to both the market impact \cite{Almgren,Bouchaud,Bouchaud2,Hasbrouck} and optimal execution \cite{Alfonsi,Almgren2,Almgren,Wang} literatures. A cost function $c_t$ and an order book $\gamma''_t$ are derived endogenously from the premises of the model. We hope that the proposed implied alpha codebook will find applications in other limit-order book models.
Finally, while we ignore inventory risk to make our market maker risk-neutral, our approach still relates to utility function and inventory models such as \cite{Amihud,Garman}.   

\vskip 2pt
The results of the paper are organized in four sections. First, the setup for the model is given, which includes a methodology for modeling heterogeneous beliefs and transaction costs. Second, a simple client model is presented and solved, leading to a relationship between trades and alphas which  motivates the market maker's choice for a codebook. Third, the market maker model is built step by step, starting from the client model and working through a series of results and approximations to lead to a reasonable control problem. Last, the market maker's problem is solved and the dynamics of the order book are determined analytically, For the sake of definiteness, we focus on a tractable example for which we  compute the two-humped limit order book shape. An appendix is added at the end of the paper to provide the proofs of two technical results which, had they been included in the text, could have distracted from the main objective of the modeling challenge.

\section{Setup of the model}

In this section, all the elements of the model are presented.

\subsection{Heterogeneous beliefs on the price}

Consider $n$ clients and one market maker who interact on an electronic exchange, with $n$ very large. We will denote by $i \in \{1...n\}$ a client index, and $k\in \{0...n\}$ a generic index, with $k=0$ corresponding to the market maker. We first introduce the following setting for the model:
\begin{enumerate}
\item a filtered probability space $(\Omega, \F, \left(\F_t\right)_{t \geq 0}, \P)$ representing the ``real life'' filtration and probability measure. The filtration is generated by a $d$-dimensional $\P$ Wiener process $W_t$.
\item a different filtration and measure $((\F^k_t)_{t \geq 0}, \P^k)$ for each of the agents. Assume furthermore that $\F^k_t \subset \F_t$, that $\P^k|_{\F^k_t}$ and $\P|_{\F^k_t}$ are equivalent and that $\P^0|_{\F^0_t} = \P|_{\F^0_t}$.\label{equivalence}
\item a $d^k$-dimensional $\P$ Wiener process $W^k_t$ that generates the filtration $(\F^k_t)_{t \geq 0}$.\label{BMs}
\item a price process $p_t$ which is an It\^{o} process adapted to \emph{all} the filtrations $\left(\left(\F^k_t\right)_{t \geq 0}\right)_{k=0...n}$.
\item the drift and volatility of $p_t$ grow at most polynomially in $t$ under all probability measures.
\end{enumerate}
where the last hypothesis must be understood in the a.s. and $L^2$ sense.

Let
\begin{equation}
dp_t = a_t dt + \sigma_t dW_t
\end{equation}
be the It\^o decomposition of $p_t$ under $\left(\left(\F_t\right)_{t\geq 0},\P\right)$. Hypotheses \ref{equivalence} and \ref{BMs} imply that there exists an $\left(\F^k_t\right)_{t\geq 0}$ adapted process $r^k_t$ such that
\begin{equation}
W^k_t = B^k_t + \int_0^t r^k_s ds
\end{equation}
for some $\P^k|_{\F^k_t}$ Wiener process $B^k_t$ and with $r^0_t = 0$. Because $W^k$ is $d_k$ dimensional, so are $r^k_t$ and $B^k_t$. Furthermore, by the martingale representation theorem, given that $W^k_t$ is a $\P$ martingale, there exists an $\left(\F_t\right)_{t \geq 0}$ adapted, $d_k \times d$ dimensional matrix $\Sigma^k_t$ such that 
\begin{equation}
dW^k_t = \Sigma^k_t dW_t
\end{equation}
Finally, agent $k$ has the following view on the market under $\left(\left(\F^k_t\right)_{t\geq 0}, \P^k\right)$:
\begin{equation}
dp_t = a^k_t dt + \sigma^k_t dB^k_t
\end{equation}
with $\sigma_t = \sigma^k_t \Sigma^{k}_t $ and $a^k_t = a_t + \sigma^k_t r^k_t$. In particular, $\sigma_t$ needs to live in the intersection of the images of all the $\left(\Sigma^{k}_t\right)^T$ for $p_t$ to be adapted to all the filtrations. Note that the $d^k$ is allowed to differ from one agent to another, in which case the $\sigma^k_t$ must be of different dimensions and hence differ.

In conclusion, all the agents have views on the price process that potentially conflict with each other's probability measure and even filtration, but can be compared coherently within the larger probability space $(\Omega, \F, \left(\F_t\right)_{t \geq 0}, \P)$.

\begin{example}
Consider the case where you have three Wiener processes $W^1$, $W^2$ and $W^3$. Let $p_t = W^1_t+W^2_t+W^3_t$, $\F^i_t = \sigma\left(\left(W^i_s\right)_{s\leq t}, \left(p_s\right)_{s\leq t}\right)$ and $r^i_t = f(W^i_t)$. You then have that:
\begin{align*}
dp_t &=  dW^1_t + dW^2_t + dW^3_t \;\textit{under}\; \P \\
     &= f(W^1_t) dt + dB^{1,1}_t + dB^{1,2}_t \;\textit{under}\; \P^1\\ 
     &= f(W^2_t) dt + dB^{2,1}_t + dB^{2,2}_t \;\textit{under}\; \P^2
\end{align*}
where the two last decompositions are not adapted with respect to the other filtration, but the price process is nevertheless adapted to \emph{both} filtrations.
\end{example}

\subsection{Transaction costs}

We now allow trades among agents. All clients have a cumulative position $L^i$ in the asset, which starts off at $0$ at the beginning of the trading period. For the market maker, we rescale all the quantities according to the number of clients, hence $L^0_t$ is his average cumulative position per client. Clients control their position through its first derivative, $l^i_t$ but incur transaction costs. The market maker, on the other hand, has no direct control over his position, but receives the liquidity fee $c_t$. To be precise, the market maker announces to his clients a cost function $l \longmapsto c_t(l)$, which denotes the price of trading at speed $l$ at moment $t$. A client then chooses her preferred trading volume $l^i_t$, and pays $c_t(l^i_t)$ in total transaction fees. We make the following hypotheses on these two processes:
\begin{enumerate}
\item The trade volume process $l^i_t$ for each client is adapted to $\left(\F^i_t\right)_{t \geq 0}$ and $\left(\F^0_t\right)_{t \geq 0}$. \label{see_trades}
\item The cost function process $c_t$ is adapted to all filtrations. \label{see_costs}
\item Marginal costs are defined: $c'_t$ is almost everywhere continuous. \label{diff_costs}
\item Clients may choose not to trade, $c_t(0) = 0$ and the mid-price is well defined at $p_t$, $c'_t(0)=0$. \label{mid-price}
\item Marginal costs increase with volume: $c_t$ is convex. \label{convex_costs}
\item There is a fixed amount of liquidity the market maker offers. \label{fixed_vol}
\end{enumerate}
Hypotheses \ref{see_trades} and \ref{see_costs} describe the information the different agents have access to. \ref{diff_costs}-\ref{convex_costs} are intuitive properties that the cost function must verify to make sense in terms of transaction costs. Finally, the hypothesis \ref{fixed_vol} is left purposefully vague, because it is best expressed with notation we introduce now.

\subsubsection{Duality and link to the order book}

This section introduces a change of variable that is both mathematically convenient and will carry deeper financial meaning when the model is solved. The new variable $\gamma$ is defined as follows:
\begin{equation}
\gamma_t(\alpha) = \sup_{l\in \text{supp}(c_t)} \left( \alpha l - c_t(l)\right) 
\end{equation}
Under the assumptions $c_t \in C^1$, convex and $c_t(0)=0$, this is equivalent to defining $\gamma'_t = (c'_t)^{-1}$ and $\gamma_t(0) = 0$. In that case, $\gamma_t$ is simply the Legendre transform of $c_t$. The following properties can be derived:
\begin{enumerate}
\item The Legendre transform maps the space of convex functions onto itself. It is its own inverse. In particular, $\gamma''_t$, where the second derivative has to be understood in the sense of distributions, is a positive, finite measure.
\item $\forall \alpha \in \mathbb{R}, \forall l\in \mathbb{R}, \gamma_t(\alpha) + c_t(l) \geq \alpha l$ (Fenchel's inequality), with equality for $\alpha^* = c'_t(l)$ or $l^* = \gamma'_t(\alpha)$. 
\item $\gamma''_t$ is a description of the limit-order book of the market maker. Indeed, it represents marginal volumes as a function of marginal costs. The property $\gamma''_t$ being a positive, finite measure corresponds to the fact that the market maker can only post positive, finite volumes on the order book.
\end{enumerate}
This means we can actually replace $c_t$ by the second derivative of its dual $\gamma''_t$ as the control variable of the market maker. Similarly, $\alpha^i_t$ is the dual to the client's control variable $l^i_t$. Call $\gamma''_t$ the market maker's liquidity offer and $\alpha^i_t$ the client's implied alpha. The second terminology will be explained later. We therefore recast all the hypotheses with these new dual variables:
\begin{enumerate}
\item The implied alpha process $\alpha^i_t$ for each client is adapted only to $\left(\F^i_t\right)_{t \geq 0}$ and $\left(\F^0_t\right)_{t \geq 0}$.
\item The liquidity process $\gamma_t''$ for each market maker is adapted to all filtrations.
\item We have $l^i_t = \gamma_t'\left(\alpha_t^i\right)$ and $\alpha^i_t = c'_t(l^i_t)$.  \label{duality_relationship}
\item Clients may choose not to trade and the order book is centered around the mid-price: $\gamma_t(0)=0$ and $\gamma'_t(0)=0$. \label{gamma_centered}
\item Only positive, finite volumes are posted on a market maker's liquidity offer: $\gamma''_t$ is a positive, finite measure. \label{orderbook_measure}
\item The total mass of $\gamma_t''$ is fixed. For convenience sake we will renormalize it to one, making $\gamma''_t$ a probability measure.
\end{enumerate}
As a consequence of \ref{duality_relationship}-\ref{orderbook_measure}, $l_t^i$ is bounded. Furthermore, $\gamma''_t$ therefore lives on the space of finite measures, which is a complete, separable metric space under the L\'evy-Prokhorov metric and a convex set.

\subsection{Putting it all together}

We summarize the model. First a public price process $p_t$ is given. Assume it is exogenously given to every one, that is, all the agents consider that they have no impact on the price under their probability measure, and only try to estimate its future movements. Then, a public liquidity offer, which can both be followed either by the cost function $c_t$ or its dual, $\gamma''_t$ is announced. Clients pick their trade volumes $l^i_t$, or equivalently, their implied alphas $\alpha^i_t$. We furthermore have the following state variable equations:
\begin{equation}
\left\{
	\begin{array}{rl}
		dL_t^i &= l^i_t dt \\
					 &= \gamma_t'\left(\alpha_t^i\right) dt \\
		dL_t^0 &= - \frac{1}{n}\sum_i \gamma_t'\left(\alpha_t^i\right) dt
	\end{array}
\right.
\end{equation}
since $\gamma'_t$ is bounded, $L_t^k$ has at most linear growth in $t$.

Finally, we write an infinite-horizon objective function for each agent, using distinct time scales $\beta^k$ for each of them. Cumulative positions appreciate or depreciate at every moment by $dp_t$ and client $i$ pays the liquidity fee $c_t(l^i_t)$ to the market maker when readjusting her portfolio. We also rescale his objective function according to the number of clients he has.
\begin{equation}
\left\{
	\begin{array}{rl}
		J^i &= \E_{\P^i}\left[\int_0^\infty e^{-\beta^i t} \left(L_t^i dp_t - c_t\left(l_t^i\right) dt\right)\right] \\ 
				&= \E_{\P^i}\left[\int_0^\infty e^{-\beta^i t} \left(L_t^i dp_t - \alpha_t^i \gamma_t'\left(\alpha_t^i\right) dt + \gamma_t\left(\alpha_t^i\right)dt\right)\right] \\
		J^0 &= \E_{\P}\left[\int_0^\infty e^{-\beta^0 t} \left(L_t^0 dp_t + \frac{1}{n}\sum_i \left(\alpha_t^i \gamma_t'\left(\alpha_t^i\right) - \gamma_t\left(\alpha_t^i\right)\right) dt\right)\right]
	\end{array}
\right.
\end{equation}
The first term of each objective function is bounded by the linear growth of $L_t^k$ and the polynomial growth of the drift and volatility of $p_t$ under $\P^k$.

\section{The Client Control Problem}

We first solve the problem from the clients' perspective, in order to know what drives trades. The results of this simple model will be used in the next section as a codebook for a more powerful market making model. The section concludes with a small robustness analysis on the underlying hypothesis of the model.

\subsection{Solving the control problem}

In this section, we solve the control problem for one generic client. Because her decisions have no impact on either $p_t$ or $c_t$, it will not affect any of the other clients' decisions. Her control variable is not even adapted to their filtration. Similarly, her own control problem is not affected by the other clients' decision. This allows us to drop the index $i$ in this subsection. Let $\left(\tilde{\P},\left(\tilde{F}_t\right)_{t\geq 0}\right)$ denote that client's probability measure and filtration and $\tilde{W}$ her Wiener process. Remember that the client tries to solve the following control problem:
\begin{itemize}
\item An admissible control $l_t$ is a stochastic process adapted to $\left(\tilde{F}_t\right)_{t\geq 0}$ that lives in the support of $c_t$. Because of the hypothesis $\left\langle\gamma''_t,1\right\rangle =1$, we have that the support of $c_t$ is included in $[-1,1]$, which means that the control set $A$ is bounded.  

\item The state variable $L_t$ verifies the dynamics:
\begin{equation}
dL_t = l_t dt
\end{equation}

\item The objective function is
\begin{equation}
\sup_l \E_{\tilde{\P}}\left[\int_0^\infty e^{-\beta t} \left(L_t dp_t - c_t\left(l_t\right)dt\right)\right]
\end{equation}
We assume $\beta$ large enough for the problem to be well defined.
\end{itemize}

Using the notation introduced in the previous section, we have the central result:
\begin{theorem}{(Implied alpha)}

A client who trades optimally will follow the relationship
\begin{equation}
\alpha_t = \E_{\tilde\P}\left[\left. \int_t^\infty e^{-\beta(s-t)} dp_s\right|\tilde{\F}_t\right] \label{implied_alpha}
\end{equation}
where $\alpha_t$ is defined by the codebook $c'_t(l_t) = \alpha_t$. 
\end{theorem}
\begin{proof}

The aim now is to apply the Pontryagin maximum principle to the above control problem. The gain function is not in standard form, but a simple integration by parts solves this issue.
\begin{equation}
\int_0^\infty e^{-\beta t} L_t dp_t = \left[ e^{-\beta t} L_t p_t\right]_0^\infty  - \int_0^\infty e^{-\beta t} l_t p_t dt + \int_0^\infty e^{-\beta t} L_t \beta p_t dt.
\end{equation}
The first term is equal to $0$ because we assume $L_0 = 0$ and $\lim_{t\rightarrow \infty} e^{-\beta t} L_t p_t = 0$ by the linear growth of $L_t$ and the polynomial growth of $p_t$. The client therefore maximizes:
\begin{equation}
 \E_{\tilde{\P}}\left[\int_0^\infty e^{-\beta t} \left( \left(\beta L_t - l_t\right) p_t -c_t\left(l_t\right)\right)dt\right]
\end{equation}
We write out the generalized Hamiltonian of the system:
\begin{equation}
	\begin{array}{rl}
		\mathcal{H}(t, \omega, L, l, Y) &= l Y +  e^{-\beta t}\left((\beta L - l) p_t - c_t(l)\right)
	\end{array}
\end{equation}
The generalized Hamiltonian is linear in $L$ and concave in $l$ by convexity of $c$, and therefore overall concave in $(L,l)$.

$p_t$ and $c_t$ are exogenously defined It\^o process. Then, the backward equation
\begin{equation}
  -dY_t =  \beta p_t e^{-\beta t} dt - Z_{t} \cdot d\tilde{W}_t
\end{equation}
has a unique solution
\begin{equation}
  Y_t = \E_{\tilde\P}\left[\left. \int_t^{\infty} \beta p_s e^{-\beta s} ds\right| \tilde{\F}_t\right]
\end{equation}
By polynomial growth of the volatility of $p_t$, the $Z$ term of the Backward Stochastic Differential Equation (BSDE from now on) satisfies the growth condition (\ref{Z_growth}).

Therefore, the candidate optimal control is $l^*_t$ verifying $e^{\beta t}Y_t - p_t = c'_t(l^*_t)$. 

This determines $l_t = \gamma'_t\left(e^{\beta t} Y_t - p_t\right)$. Therefore the forward equation of $L_t$ has a unique solution.

Finally, the quantity
\begin{align}
c'_t(l^*_t) &=  e^{\beta t} Y_{t} - p_t = \E_{\tilde\P}\Bigg[\left. \int^\infty_t \beta p_s e^{-\beta(s-t)} ds - p_t\right|\tilde{\F}_t\Bigg] \\
					&= \mathbb{E}_{\tilde\P}\Bigg[\left. \int_t^\infty e^{-\beta(s-t)} dp_s\right|\tilde{\F}_t\Bigg]
\end{align}
can be seen as the ``implied alpha'' of the deal, that is, the difference between the projected value into the future and the current value of $p$.
\end{proof}

While it makes intuitive sense\footnote{All the formula says is that marginal costs equal expected marginal gains.}, this a non-trivial result.  Indeed, $c_t$ and hence $l_t$ depend on the market maker's decision, and yet $\alpha_t = c'_t\left(l_t\right)$ becomes a quantity that is ``intrinsic'' to the client. It is independent of the market maker's pricing and only depends on the client's view on the market. It intuitively represents the client's price estimator and summarizes her beliefs on the price dynamics.  Note that the discount factor now becomes the time-scale of her prediction. We define the quantity the client tries to predict as the \emph{realized alpha} over the time scale $\frac{1}{\beta}$:
\begin{equation}
\alpha^r_t = \int_t^\infty e^{-\beta(s-t)} dp_s.
\end{equation}
This coincides with what practitioners refer to as the 'alpha' of a client.

\subsection{Dynamics of the implied alpha}
(\ref{implied_alpha}) can be rewritten in terms of It\^o dynamics:
\begin{align}
d\alpha_t &= \beta \alpha_t dt - dp_t + e^{\beta t} Z_t d\tilde{W}_t \\
					&= \beta \alpha_t dt - dp_t +  \theta_t d\tilde{W}_t \label{central}
\end{align}
with $\theta_t = e^{\beta t} Z_t $ therefore being the volatility of the estimation. 
All these equations summarize the link between optimal client volumes and price dynamics \emph{under the client's probability measure and filtration}. 

Three things can be noted. 
\begin{enumerate}
\item The drift of the implied alpha is the result of two opposing forces. On the one hand, the self-correlation term guarantees a certain coherence in the client's decisions over the time-scale $\beta^{-1}$. On the other hand, the implied alpha, through the $-dp_t$ term, automatically takes into account the last price variation to recenter the estimation. 
\item $\theta_t$ is a measure of intelligence of a client over the price process $p_t$, given that in the limit where $\theta_t = 0$, a client has a perfect view on the market. Conversely, clients with a big $\theta_t$ will have a higher variance on their price estimator, and can at the limit be considered as ``noise'' traders. 
\item We can write the dynamics of $\alpha_t$ under $\P$ to obtain:
\begin{equation}
d\alpha_t = \beta \alpha_t dt  - dp_t + \theta_t \Sigma_t dW_t + \theta_t r_t dt \label{multi_central}
\end{equation}
which provides the link between trade and price dynamics. Note that, while under $\tilde{\P}$, $\alpha_t$ is intrinsic to the client, under $\P$, $\alpha_t$ may depend on the market maker's decisions. In words, this means that while the market maker cannot influence the client's decision under her own view of the market, he can affect that view itself.
\end{enumerate}

\subsubsection{The cost of information}

Given the above result, it is clear that a market maker has a privileged position on the market: he catches a glimpse of everyone's belief on the price. We can now give an interpretation of $\gamma$ beyond the fact that $\gamma''$ represents the order book:

Under a martingale measure of $p_t$, $c_t(l_t)$ represents the cost the client pays to the market maker for the liquidity $l_t$, and this is the conservative interpretation of transaction costs. However, $\gamma_t(\alpha_t)$ represents the cost the market maker pays \emph{under the client's view of the market}. Under $\tilde{\P}$, it is as if the market maker pays the client for information on the price process.

This gives us a good intuition about the market maker's strategy: he collects information from each client, pricing them according to the current beliefs and how much the new information brings to him:  $\gamma(\alpha)$ is essentially how much the market maker is willing to pay for a prediction of strength $\alpha$, assuming that the client is correct.

\subsection{Robustness analysis}

In this section we define the notion of a 'cash-insensitive' agent and generalize the implied alpha relationship to the utility function case. This illustrates the robustness and limits of the proposed codebook.

Define the state variables
\begin{equation}
\left\{
	\begin{array}{rl}
		dL_t &= l_t dt \\
		dK_t &= - \left(p_t l_t + c_t(l_t) \right) dt
	\end{array}
\right.
\end{equation}
The second variable represents the client's \emph{cash position} at time $t$. If we define her wealth process as $X_t = p_t L_t + K_t$ then it verifies the standard dynamics
\begin{equation}
dX_t = L_t dp_t
\end{equation}
The proposed decomposition has a clear economic meaning: the agent calculates her wealth by adding her cash $K_t$ and the marked-to-the-mid value of her asset position $L_t$. This means that a 'cash-insensitive' objective function would be of the standard form
\begin{equation}
 J = \E\left[U(X_\tau,p_\tau)\right] = \E\left[U(p_\tau L_\tau + K_\tau,p_\tau)\right]
\end{equation}
with $\tau$ a stopping time adapted to the client's filtration and $U$ her utility function. The special form guarantees that the client does not differentiate between wealth in cash and wealth in the asset. An agent concerned with the liquidity of the asset would not be 'cash-insensitive'. 

One could generalize the utility framework to 'cash-sensitive' agents by considering a utility function of the form $U(L_\tau, K_\tau, p_\tau)$. In this case, the below result would not hold.
The new Hamiltonian becomes:
\begin{equation}
\mathcal{H}(t,\omega, L, K, l, Y_L, Y_K) = l Y_L - (p_t l + c_t(l)) Y_K
\end{equation}
where the dual variables satisfy in the 'cash-insensitive' case the BSDEs
\begin{align*}
Y_{L,t} &= \E\left[\left. \partial_X U(X_\tau,p_\tau) p_\tau \right|\F_t\right] \\
Y_{K,t} &= \E\left[\left. \partial_X U(X_\tau,p_\tau) \right|\F_t\right]
\end{align*}
The optimal execution strategy therefore verifies:
\begin{equation}
c'_t(l^*_t) = \E\left[\left. \frac{Y_{K,\tau}}{Y_{K,t}} p_\tau \right|\F_t\right] - p_t
\end{equation}
which can be rewritten in a manner very similar to (\ref{implied_alpha}):
\begin{equation}
\alpha_t = \E_\mathbb{Q} \left[\left. p_\tau \right|\F_t\right] - p_t
\end{equation}
where $\frac{d\mathbb{Q}}{d\P} = \frac{Y_{K,\tau}}{\E\left[Y_{X,\tau}\right]}$ is a legitimate change of measure if $\partial_X U(X_\tau,p_\tau)$ is positive and integrable. In the case where $\tau$ is exponential and independent of the price process, we simply recover (\ref{implied_alpha}) under a different probability measure. Given that we assume all the clients to have differing probability measures anyway, we can without loss of generality stick to (\ref{implied_alpha}).

The above computation is somewhat formal, but can be made rigorous by giving explicit integrability assumptions on $\tau$ such that the growth condition (\ref{growth}) is verified. This is in particular the case when $\tau$ is independent of $p$ and has exponential moments, or if $\tau$ is bounded.

\section{Reworking the Market Maker's <odel}

In this section, we work under the measure $\P$  and often drop the superscrip $0$ referring to the market maker. Our goal is to provide an optimal market making strategy. Because of its complexity, the problem cannot be solved in full generality, and we propose a set of approximations which we justify on financial grounds.

The next four subsections identify a succession of simplifications guided by intuition based on the behavior of a typical market maker:
\begin{enumerate}
\item First, he should 'not hold a view on the market'. Mathematically speaking, this means that in his model for the price, the main explanatory variables are his client's beliefs. The error associated to this first simplification is proved to be small, though a function of the market maker's control.
\item Second, he should model how his clients' views might evolve into the future. A straightforward system of correlated Ornstein-Uhlenbeck processes is proposed to serve this purpose. This will be used to define an approximate objective function for the market maker.
\item The third approximation is made for mathematical convenience: we assume that the market maker has an infinite number of clients. This leads to the previous models becoming SPDEs.
\item Finally, a placeholder function is proposed to model the source of error identified in the first subsection.
\end{enumerate}

\subsection{Approximate Price Process}

As the market maker should not hold a view on the market, we refrain from directly modeling $p_t$ under the market maker's measure and filtration. Instead, the market maker constructs his model from the client's implied alphas. This is done in the following fashion:

Equation (\ref{multi_central}) can be turned around to describe price dynamics using the implied alpha of client $i$ under $\P$:
\begin{equation}
dp_t = \beta^i \alpha^i_t dt  - d\alpha^i_t + \theta^i_t \Sigma^i_t dW_t + \theta^i_t r^i_t dt
\end{equation}
and this equation holds true for all $i$. Notice that there is no contradiction with the uniqueness of the It\^o decomposition: these representations correspond to the same It\^o process rewritten in terms of the variable $\alpha^i_t$. Using a sequence of positive and constant weights $\lambda^i$ averaging to $1$ (i.e. such that $(1/n)\sum_{i=1}^n \lambda^i = 1$), we obtain:
\begin{equation}
dp_t = \frac{1}{n}\sum_{i=1}^n \lambda^i \left( \beta^i \alpha^i_t dt  - d\alpha^i_t \right) +  \left(\frac{1}{n}\sum_{i=1}^n \lambda^i \theta^i_t \Sigma^i_t \right) dW_t + \left(\frac{1}{n}\sum_{i=1}^n \lambda^i \theta^i_t r^i_t\right)dt \label{estimation}
\end{equation}
Again, this equation is just a reformulation of the It\^o decomposition of $p_t$. The advantage of that rparticular epresentation is that the first term is adapted to $\left(\F^0_t\right)_{t \geq 0}$. Hence, if he also knows (or rather, chooses) the weights $\left(\lambda^i\right)_{i=1...n}$, then the market maker can follow the first term of this decomposition \emph{in real time}. We introduce the special notation $p^\lambda_t$ for this term and we refer to it as his price estimator: So:
\begin{equation}
\label{approx}
dp_t^\lambda = \frac{1}{n}\sum_{i=1}^n \lambda^i \left( \beta^i \alpha_t^i dt  - d\alpha^i_t \right) .
\end{equation}
The remainder which we denote $\epsilon_t^\lambda$ includes quantities that are unknown to him since
\begin{equation}
d\epsilon^\lambda = \left(\frac{1}{n}\sum_{i=1}^n \lambda^i \theta^i_t \Sigma^i_t \right) dW_t + \left(\frac{1}{n}\sum_{i=1}^n \lambda^i \theta^i_t r^i_t\right)dt 
\end{equation}
If we replace $p$ by $p^\lambda$ in the market maker's problem, the only difference appears in the objective function, which now contains an extra term. We refer to is as the error term:
\begin{equation}
\text{err} = \mathbb{E}\left[ \int_0^\infty e^{-\beta t} L_t \left(\frac{1}{n}\sum_{i=1}^n \lambda^i \theta^i_t r^i_t\right)dt \right]
\end{equation}
on the market maker's objective function. Next we define 
\begin{equation}
\left(\sigma^i\right)^2 = \E\left[\int_0^\infty e^{-\beta t} |\theta^i_t|^2 dt\right].
\end{equation}
This quantity is a measure of how 'noisy' client $i$ is, and can be estimated by fitting the expression for the implied alpha to historical data. To be more specific, recall that under the client measure and filtration,  
\begin{equation}
 \alpha^i_t = \mathbb{E}_{P^i}\left[\left. \int_t^{\infty} e^{-\beta^i (s-t)} dp_s \right| \mathcal{F}^i_t \right]
\end{equation}
and that $\theta^i_t $ is the error on this estimation. Hence, the smaller $\theta^i_t $, the closer the implied alpha is to the realized one, which means that he client is particularly well informed. Note that in the finance literature, clients are often called informed if their market impact function (their average alpha) is unusually large. Our notion of intelligence of the price process does not coincide with this practice, as one can have at the same time a systematically small though correct alpha. In this paper, an informed trader is a trader for whom the implied and realized alphas nearly coincide, whereas for a noise trader, relationship (\ref{implied_alpha}) has much higher variance (due, for example, to liquidity concerns, a strongly non-linear utility function, or a poor filtration). The $\beta^i$ and $\sigma^i$ can be estimated from historical data by regressing the implied alpha against the realized one, that is for example by solving the least squares regression problem
\begin{equation}
\inf_{\beta>0} \frac{1}{N} \sum_{k=1}^N \left(\alpha^i_{t_k} - \int_{t_k}^{\infty} e^{-\beta (s-{t_k})} dp_s\right)^2
\end{equation}
where the $t_k$ are the times in the past at which client $i$ traded.

A skilled market maker can therefore construct his approximate price process by choosing a $\lambda$ which puts most of its weight on such intelligent clients and little weight on \emph{noise traders} with large $\sigma^i$'s.

Using Cauchy-Schwartz's inequality and the particular choice of weights 
\begin{equation}
\lambda^i = \frac{n \left(\sigma^i\right)^{-2}}{\sum_j \left(\sigma^j\right)^{-2}},
\end{equation}
and assuming $L_t$ is uniformly  bounded by a constant $\overline{L}$, we obtain: 
\begin{align*}
|\text{err}|^2  &\leq  \beta^{-1} (\overline{L})^2    \E\left| \int_0^\infty \beta e^{- \beta t} \frac{1}{n} \sum_{i=1}^n \lambda^i  \theta^i_t r^i_t dt \right|^2 \\
							&\leq  \beta (\overline{L})^2  \frac{1}{n} \sum_{i=1}^n \int_0^\infty e^{- \beta t} \E|\lambda^i \theta_t^i|^2  \frac{1}{n}  \int_0^\infty e^{- \beta t}\sum_{i=1}^n\E\left|r^i_t\right|^2 dt \\
							&\leq \beta I^{-1} (\overline{L})^2    \int_0^\infty e^{- \beta t}  \frac{1}{n}\sum_{i=1}^n \E\left|r^i_t\right|^2 dt \\
\end{align*}
where 
\begin{equation}
I^{-1} = \min_\lambda  \frac{1}{n} \sum_{i=1}^n \left(\lambda^i \sigma^i\right)^2 =  \left(\frac{1}{n}\sum^n_{i=1} \left(\sigma^i\right)^{-2}\right)^{-1}.
\end{equation}
$I$ can therefore be seen as a measure of the aggregate intelligence the clients have over the price process. Note that it suffices for \emph{one} $\sigma^i$ to be of order $\epsilon$ for $I$ to be of order $\epsilon^{-1}$. By Girsanov's theorem, we then have that
\begin{equation}
\E\left|r^i_t\right|^2 = -2 \frac{d}{dt} \E \log\frac{d\P^i|_{\F^i_t}}{d\P|_{\F^i_ t}},
\end{equation}
and finally, 
\begin{equation}
|\text{err}|^2 \leq 2 \beta^2 I^{-1} (\overline{L})^2  \int_0^\infty e^{- \beta t} \frac{1}{n}\sum_{i=1}^n \E \log\frac{d\P^i|_{\F^i_t}}{d\P|_{\F^i_t}} dt \label{err_bound}
\end{equation}
Two important remarks are in order at this point. First, as long as at least \emph{one} agent is well informed, $I^{-1}$ is small. Second, the dependence of $p$ and $\alpha^i$ upon the market maker's control $\gamma''$ is hidden in the Radon Nykodym derivative $\frac{d\P^i|_{\F^i_t}}{d\P|_{\F^i_t}}$. To understand why it is (unfortunately) reasonable to assume that this term may be strongly dependent upon $\gamma''$ is because a client can use the information on the order book $\gamma''$ publicly available as one of the sources of information he uses to form his probability measure on the price. This problem will be addressed in the last subsection. 

\subsection{Approximate Objective Function}

In this subsection, we take a crucial methodological step. Instead of maximizing 
\begin{equation}
	J = \E\left[\int_0^\infty e^{-\beta t} \left(L_t dp_t + \frac{1}{n}\sum_i \alpha_t^i \gamma_t'\left(\alpha_t^i\right) dt - \frac1n\sum_i\gamma_t\left(\alpha_t^i\right)dt\right)\right]
\end{equation}
we assume that the market maker maximizes the \emph{approximate} objective function
\begin{equation}
	J^\lambda = \int_0^\infty e^{-\beta t} \frac{1}{n}\sum_{i=1}^n  \mathbb{E}\left[L_t \lambda^i (\beta^i - \beta) \alpha^i_t + \left(\alpha^i_t - \frac{1}{n}\sum_{j=1}^n \lambda^j \alpha^j_t\right) \gamma_t'(\alpha^i_t) - \gamma_t(\alpha^i_t)\right]dt
\end{equation} 
subject to a constraint of the form $|\text{err}|^2 \leq C$ for some constant $C>0$. The new objective function was obtained by replacing $dp_t$ by $dp_t^\lambda$ and integrating by parts.
Because of our choice of the form of the approximate price (\ref{approx}), the market maker's objective function does not depend upon $p_t$ anymore, and \emph{he only needs to model the client belief distribution}. This is consistent with the intuition that a market maker should not hold a view on the market. Rather, he should model the behavior of his clients with respect to each other, and price according to the information they provide. This is exactly the approach used in what follows.

However, we still need to propose a model for the $\alpha^i_t$. For reasons of tractability we choose them as correlated Ornstein-Uhlenbeck processes:
\begin{equation}
	d\alpha^i_t = - \rho \alpha^i_t dt + \sigma dM^i_t + \nu dW_t
\end{equation}
with $\rho >0$ and the $M^i$ Wiener processes which are independent of each other and of $W$. $\sqrt{\sigma^2+\nu^2}$ is the overall volatility level of a client, and $\nu$ the volatility that is due to some common information amongst client s(for example, the movement of the midprice).
The next step of our strategy is to introduce a penalization term to account for the possible feedback effects hidden inside the error term introduced when replacing the objective function by its approximation. But first, we take the limit $n \rightarrow \infty$ to identify effective equations providing informative approximation to the properties of the original system comprising \emph{finitely many clients}.

\subsection{Infinitely Many Clients}
Assume that the number of clients of the market maker is large enough to justify an approximation in the asymptotic regime $n$ large. This will greatly improve the tractability of the model by allowing us to work in function spaces and rely on stochastic calculus tools to solve the model.
The mainstay of this subsection is the Stochastic Partial Differential Equation (SPDE for short):
\begin{equation}
dv_t = \left(\frac{1}{2}(\sigma^2 + \nu^2) \Delta v_t + \rho \nabla\left(\text{id }  v_t\right)\right)dt - \nu \nabla v_t dW_t \label{OU_SPDE}
\end{equation}
describing the dynamics of an infinite dimensional measure valued Ornstein-Uhlenbeck  process. The following lemma links this macroscopic SPDE to our microscopic Orstein-Uhlenbeck model for the implied alphas.

\begin{proposition}
\label{pr:spde}
If $(\epsilon^i)_{i\ge 1}$ is a sequence of random variables such that $(\alpha^i_0,\epsilon^i)_{i\ge 1}$ is an iid sequence independent of $W$ and the sequence $(M^i)_{i\ge 1}$, and such that the joint distribution, say $m$, of all the couples $(\alpha^i_0,\epsilon^i)$ satisfies:
\begin{equation}
\label{of:m}
\int (|\alpha|^p+|\epsilon|^2) m(d\alpha,d\epsilon)<\infty,
\end{equation}
for all $p>0$, then for each $t\ge 0$, the limit 
\begin{equation}
\nu_t = \lim_{n \rightarrow \infty} \frac{1}{n} \sum_{i=1}^n \epsilon^i \delta_{\alpha^i_t}
\end{equation}
exists almost surely in the sense of weak convergence of measures, almost surely for every $t\ge 0$ it holds:
\begin{equation}
\label{of:integrability}
\int |\alpha|^p\;d\nu_t(\alpha)<\infty
\end{equation}
for every $p>0$, and the measure valued process $(\nu_t)_{t\ge 0}$ is a weak solution of the SPDE (\ref{OU_SPDE}) in the sense that
for any twice continuously differentiable function $f$ (i.e. $f\in C^2$) such as $f$ and its two derivatives have at most polynomial growth, we have that
\begin{equation}
d\left\langle f, v_t\right\rangle = \left\langle \frac{1}{2}(\sigma^2 + \nu^2)\Delta f -\rho id \nabla f,v_t\right\rangle dt + \nu \left\langle\nabla f,v_t\right\rangle dW_t
\end{equation}
Furthermore, for each $t>0$, the measure $v_t$ possesses an $L^2$ density almost surely.
\end{proposition}\begin{proof}
See appendix.
\end{proof}
The papers \cite{Kurtz2, Kurtz} provide similar results for a more general class of microscopic models, including existence and uniqueness of the solution of the SPDE (\ref{OU_SPDE}). However, given the simple form of the dynamics chosen in our particular model, we can provide an \emph{explicit} form for the solution and detailed properties on the nature of its tails. Note that the correlation between client beliefs is crucial in having a fully stochastic model, given that for $\nu =0$, $\nu_t$ only satisfies a deterministic partial differential equation.

\vskip 2pt
We shall use four measure valued solutions of the above SPDE. In each case, the weights $\epsilon^i$ are explicit functions of the parameters $\sigma^i$ and $\beta^i$ introduced earlier. Because the values of these parameters appear as outcomes of statistical estimation procedures in practice, assuming that they are random and satisfy some form of ergodicity is not restrictive\footnote{The fact that we have to enlarge the Brownian filtration at time $t=0$ to \emph{randomize} the coefficients does not impact the martingale representation theorem.}. To construct these measures we assume that $(\alpha^i_0)_{i\ge 1}$ is an iid sequence of random variables whose common distribution has finite moments of all orders. The first of our four measures is obtained by choosing $\epsilon^i\equiv 1$ for all $i\ge 1$. Then:
\begin{equation}
\label{of:mut}
\mu_t = \lim_{n \rightarrow \infty} \frac{1}{n} \sum_{i=1}^n \delta_{\alpha^i_t}.
\end{equation}
Next we assume that $\{(\sigma^i)^{-2}\}_{i\ge 1}$ is a sequence of positive random variables of order $1$ such that 
$\{(\alpha^i_0,(\sigma^i)^{-2})\}_{i\ge 1}$
is an iid sequence independent of $W$ and the sequence $(M^i)_{i\ge 1}$, so by choosing $\epsilon^i=(\sigma^i)^{-2}$ for all $i\ge 1$ we can define: 
\begin{equation}
\label{of:sigmat}
I_t = \lim_{n \rightarrow \infty}\frac{1}{n} \sum_{i=1}^n \left(\sigma^i\right)^{-2} \delta_{\alpha^i_t}.
\end{equation}
We shall also assume that the number $I=\mathbb{E}[(\sigma^i)^{-2}]$ is finite and strictly positive, and for each $t\ge 0$ we define the non-negative measure $\lambda_t$ by:
\begin{equation}
\label{of:lambdat}
\lambda_t = I^{-1}_0 I_t.
\end{equation}
Finally, we assume that $(\beta^i)_{i\ge 1}$ is a sequence of bounded random variables such that $\{(\alpha^i_0,(\sigma^i)^{-2},\beta^i)\}_{i\ge 1}$
is an iid sequence independent of $W$ and the sequence $(M^i)_{i\ge 1}$, so by choosing $\epsilon^i=\beta^i(\sigma^i)^{-2}$ for all $i\ge 1$ we can define: 
\begin{equation}
\label{of:betat}
\beta_t = I^{-1}\lim_{n \rightarrow \infty}\frac{1}{n} \sum_{i=1}^n \beta^i \left(\sigma^i\right)^{-2} \delta_{\alpha^i_t}.
\end{equation}
We now make a few remarks on the properties shared by essentially all the solutions  $(v_t)_{t\ge 0}$ of the SPDE (\ref{OU_SPDE}). For the sake of definiteness, we shall assume that $(v_t)_{t\ge 0}$ is a non-negative measure valued process solving this SPDE.
\vskip1pt\noindent
(1) The total mass $\langle 1,v_t\rangle$ of the measure $v_t$ is constant over time. Indeed, using the constant test function $f\equiv 1$ in (\ref{OU_SPDE})
we see that 
\begin{equation}
d\left\langle 1,v_t\right\rangle = 0 \label{SPDE_moment0}
\end{equation}
In particular, the intelligence assumption $I = \left\langle 1, I_0\right\rangle \ge \epsilon^{-1}$ is conserved over time.
\vskip1pt\noindent
(2) If we use the test function $f=id$ in (\ref{OU_SPDE})
where the identity function $id$ is defined by $id(\alpha)=\alpha$, then we see that the first moment of $v_t$ is itself an Ornstein-Uhlenbeck process mean reverting around $0$ since: 
\begin{equation}
 \label{SPDE_moment1}
d\left\langle id,v_t\right\rangle = -  \rho \left\langle id, v_t \right\rangle dt + \nu \left\langle 1, v_0\right\rangle dW_t.
\end{equation}
\vskip1pt\noindent
(3) Using $f=id^2$ we see that the second moment mean reverts around $\left(\sigma^2+\nu^2\right) \left\langle 1, v_0\right\rangle$ since:
\begin{equation}
d\left\langle id^2,v_t\right\rangle = \left(\left\langle\sigma^2 + \nu^2, v_0\right\rangle-  2 \rho \left\langle id^2, v_t \right\rangle\right) dt + 2 \nu \left\langle id, v_t \right\rangle dW_t.
\end{equation}
These Stochastic Differential Equations (SDEs for short) guarantee the existence of a constant $C_1$ such that:
\begin{align*}
\E\left\langle id^2,v_t\right\rangle &\leq e^{C_1 t} \\
\E\left\langle |id|,v_t\right\rangle &\leq \sqrt{\E\left\langle id^2,v_t\right\rangle} \leq e^{\frac{1}{2}C_1 t}.
\end{align*}
We shall use these estimates for the measures $\mu_t$, $\lambda_t$ and $\beta_t$. In fact similar estimates hold for moments of all order as can be proved by induction from (\ref{OU_SPDE}). We do not give the details as they will not be used in what follows.

\vskip 4pt
Coming back to the optimal control problem of the market maker, since $\gamma''_t$ belongs to a space of probability measures whenever the control $(\gamma_t)_{t\ge 0}$ is admissible,  the following estimates hold: 
\begin{align*}
||\gamma'_t||_\infty &\leq 1 \\ 
|\gamma_t(\alpha)| &\leq \alpha \\
\left|\left\langle \gamma'_t,v_t\right\rangle\right| &\leq \left\langle 1, v_0\right\rangle = O(1) \\
\E\left|\left\langle \gamma_t, v_t\right\rangle\right| &\leq \E\left\langle |id|,v_t\right\rangle = o\left(e^{\beta t}\right) \\
\E\left|\left\langle \text{id } \gamma'_t,v_t\right\rangle\right| &\leq \E\left\langle |id|, v_t \right\rangle = o\left(e^{\beta t}\right)
\end{align*}

\vskip 4pt
As an immediate consequence of the above remarks we have:
\begin{corollary}
For any progressively measurable process $(\gamma_t)_{t\ge 0}$ such that $\gamma_t''$ is a probability measure, the state dynamic equation of the market maker:
\begin{equation}
\label{fo:mmdynamics}
	dL_t = - \left\langle \gamma'_t, \mu_t \right\rangle dt
\end{equation}
makes sense, and for sufficiently large $\beta$, the approximate objective function
\begin{equation}
	J^\lambda = \int_0^\infty e^{-\beta t}  \mathbb{E}\left[L_t \left\langle id , \beta_t\right\rangle  + \left\langle - L_t \beta id + \left( id - \bar{\alpha}_t \right) \gamma'_t - \gamma_t ,\mu_t \right\rangle\right] dt \label{limit_objective}
\end{equation}
where $\bar{\alpha}_t = \left\langle id,\lambda_t\right\rangle$, is well defined.
\end{corollary}

\subsection{Modeling the Error Term}

As argued in the first subsection, the approximation technique hinges on one hypothesis:  the existence of at least one '$\epsilon$-intelligent' client. Furthermore, the clients probability measure (by which we mean the distribution of the clients implied alphas) is potentially a function of the market maker's control. This causes an undesirable nonlinear feedback effect which needs to be reined in to avoid explosion of the approximation error. In this subsection we propose a direct description of the error term, leaving open the question of how to derive it from a specific model of the clients probability measures.

Intuitively, the feedback effect corresponds to how much new information the order book shape reveals to the clients. The clients' beliefs at time $t$ can be summarized by the probability measure $\mu_t$ and the order book by $\gamma''_t$. A reasonable model for the error term is given by the expression
\begin{equation}
\label{penalization}
E = \epsilon \E\int_0^\infty e^{-\beta t} H(\gamma''_t|\mu_t) dt
\end{equation}
where $H(\nu|\mu)$ denotes the Kullback-Leibler distance (also known as relative entropy) defined by
\begin{equation}
\label{fo:KL}
H(\nu|\mu)=
\begin{cases}
&\int f\left(\frac{d\nu}{d\mu}\right)d\mu, \;\text{whenever}\; \nu<<\mu;\\
&\infty \;\text{otherwise}
\end{cases}
\end{equation}
with $f(x) = x \log x$ . Note that $H(\nu|\mu)$ is minimal for $\gamma''_t$ = $\mu_t$ by convexity of $f$. As explained in the introduction, we choose this particular distance for its intuitive interpretation and the fact that it leads to an explicit expression for the two hump order book shape endogenous to the model. However, the results hold for general strictly convex functions $f$, in which case the pseudo-distance $H$ defined in (\ref{fo:KL}) is known as the f-divergence between the measures $\mu$ and $\nu$. See \cite{Csiszar}.

\section{The Market Maker's Control Problem}
\label{se:mm}

With all the pieces of our model in place, we now solve the market maker's control problem.

\subsection{Model Summary}
We consider a sequence $\left(\beta^i,(\sigma^i)^{-2}\right)_{i\ge 1}$ of random variables such that the assumptions of 
Proposition \ref{pr:spde} are satisfied with $\epsilon^i\equiv 1$, $\epsilon^i=(\sigma^i)^{-2}$ and $\epsilon^1=\beta^i(\sigma^i)^{-2}$ respectively,
so that the measure-valued processes $(\mu_t)_{t\ge 0}$, $(\lambda_t)_{t\ge 0}$ and $(\beta_t)_{t\ge 0}$ constructed in (\ref{of:mut}), (\ref{of:lambdat}) and (\ref{of:betat}) are well defined.

As explained earlier, the market maker's control at time $t$ is the convex function $\gamma_t$, and given our choice of penalizing terms, 
we expect its second derivative $\gamma_t''$ to be a probability measure absolutely continuous with respect to $\mu_t$ in order to avoid infinite penalties. So we refine the definition of the set ${\mathcal A}$ of admissible controls for the market maker as the set of random fields $(g_t(x))_{t\ge 0, x\in{\mathbb R}}$ such that $(g_t(x))_{t\ge 0}$ is progressively measurable for each $x\in{\mathbb R}$ fixed, and $\left\langle g, \mu_t\right\rangle =1$ for each  $t\ge 0$. Given an admissible control $g\in{\mathcal A}$, we define the function $\gamma_t$ as the anti-derivative of the function $\gamma'(\alpha) = \left\langle g_t 1_{[0,\alpha]}, \mu_t\right\rangle$ satisfying $\gamma_t(0)=0$. 
The objective of the market maker is to choose an admissible control in order to maximize his modified \emph{objective function} defined as:
\begin{equation}
\int_0^\infty e^{-\beta t}  \mathbb{E}\left[L_t \left\langle id , \beta_t\right\rangle  + \left\langle - L_t \beta id + \left( id - \bar{\alpha}_t \right) \gamma'_t - \gamma_t  - \epsilon  f\circ g_t ,\mu_t \right\rangle\right] dt
\end{equation}
where the controlled dynamics of the state $L_t$ of the system are given by:
\begin{equation}
\label{fo:statedynamics}
dL_t = - \left\langle \gamma'_t, \mu_t \right\rangle dt .
\end{equation}
We used the functions $\gamma_t$ and $\gamma_t'$ for consistency with earlier discussions. See below for forms of the state dynamics and market maker objective function written exclusively in terms of the density $g_t$). The above quantities are well-defined because of the estimates proven for the measures $\mu_t$, $\lambda_t$ and $\beta_t$ solving the SPDE.
 
\subsection{Solving the Market Maker Control Problem}
We denote by $Y$ the adjoint variable of $L$ so that the generalized (random) Hamiltonian of the problem can be defined as
\begin{equation}
\label{fo:hamiltonian}
H(t, L, g, Y)=
- \left\langle \gamma_t' , \mu_t \right\rangle Y + e^{-\beta t}  \left(L \left\langle id, \beta_t\right\rangle + \left\langle - \beta L id + (id- \bar{\alpha}_t)\gamma_t' - \gamma_t - \epsilon f \circ g, \mu_t\right\rangle \right)
\end{equation}
for any deterministic function $g$ as long as we define $\gamma$ and $\gamma'$ by $\gamma'(\alpha) = \left\langle g_t 1_{[0,\alpha]}, \mu_t\right\rangle$ and $\gamma(\alpha)=\int_0^\alpha\gamma'(\tilde\alpha)d\tilde\alpha$. Note that in the above expression of the Hamiltonian, $g$ is determinist and does not depend upon $t$, but $\gamma_t'$ which is the cumulative distribution function of the probability measure $\gamma_t''$ with density $g$ with respect to $\mu_t$ is random and depends upon $t$. Clearly, so is $\gamma_t$.
This Hamiltonian can be rewritten as
\begin{equation}
\label{fo:hamiltonian'}
H(t, L, g, Y)=
e^{-\beta t}  L \langle id, \beta_t -\beta  \mu_t\rangle +e^{-\beta t}\tilde H(g)
\end{equation}
where the modified Hamiltonian $\tilde H$ is defined by
$$
\tilde H(g)=- \langle \gamma_t' , \mu_t \rangle e^{\beta t}Y +  \langle  (id- \bar{\alpha}_t)\gamma_t' - \gamma_t - \epsilon f \circ g, \mu_t\rangle,
$$
and the other terms do not depend upon the control. Since the stochastic maximum principle proven in appendix says that we can look for the optimal control by maximizing the Hamiltonian, we shall maximize the modified Hamiltonian. For each choice of the admissible control $(g_t)_{t\ge 0}$, we consider the corresponding state process $(L_t)_{t\ge 0}$ given by (\ref{fo:statedynamics}) and the adjoint equation:
\begin{equation}
\label{fo:adjoint}
-dY_t = e^{-\beta t}\left\langle id,\beta_t - \beta \mu_t\right\rangle dt - Z_t dW_t.
\end{equation}
Since the derivative of the Hamiltonian with respect to the state variable $L$ does not depend upon the control or the state $L$,
the adjoint process can be determined independently of the choice of the admissible control $(g_t)_{t\ge 0}$ and the associated state $(L_t)_{t\ge 0}$ solving (\ref{fo:statedynamics}). Given the explicit formulas we have derived for $\mu_t$, we obtain:

\begin{lemma}\label{BSDE_MM}
The solution to the adjoint Backward Stochastic Differential Equation (\ref{fo:adjoint}) is given by:
\begin{equation}
\label{fo:yoft}
Y_t  = \frac{e^{-\beta t }}{\beta + \rho}  \left\langle id,\beta_t - \beta \mu_t\right\rangle,
\end{equation}
and it verifies the growth condition (\ref{Z_growth}).
\end{lemma}
\begin{proof}
The exact form (\ref{fo:yoft}) of the solution can be guessed by going back to the explicit system of finitely many Ornstein-Uhlenbeck processes and taking the limit. However, for the proof, we show that the process $(Y_t)_{t\ge 0}$ given by (\ref{fo:yoft}) is the solution by direct inspection, computing the It\^o differential of $Y_t$ defined by (\ref{fo:yoft}) and using the fact that the random measures $\beta_t$ and $\beta \mu_t$ also solve the SPDE (\ref{OU_SPDE}). Using (\ref{SPDE_moment1}), we get:
\begin{align*}
dY_t &=  - \beta Y_t dt - \rho Y_t dt +  e^{-\beta t } \frac{\nu}{\beta + \rho} dW_t \\
		 &=  - e^{-\beta t} \left\langle id,\beta_t - \beta \mu_t\right\rangle dt + e^{-\beta t } \frac{\nu}{\beta + \rho} dW_t
\end{align*}
and by the growth properties of the first moment, $\lim_{t\to\infty}Y_t = 0$. Moreover, $Z_t = e^{-\beta t} \nu/(\beta + \rho)$ clearly verifies the growth condition (\ref{Z_growth}).
\end{proof}
\vskip 2pt\noindent
The form of the modified Hamiltonian justifies the introduction of the quantity:
\begin{equation}
\label{fo:alpha*}
\alpha^*_t = \bar{\alpha}_t + e^{\beta t} Y_t = \left\langle id,\lambda_t + \frac{\beta_t - \beta \mu_t}{\beta + \rho}\right\rangle
\end{equation}
so that, if we compute the modified Hamiltonian along the path of the adjoint process we get:
$$
\tilde H(g)=- \langle \alpha_t^*\gamma_t' -\gamma_t, \mu_t \rangle  - \epsilon \langle f \circ g, \mu_t\rangle.
$$
$\alpha^*_t $ can be viewed as the \emph{shadow alpha} of the market maker.\footnote{This terminology is justified by the fact that, when the market maker \emph{does} have his own view on the market (that is, his own model for $dp_t$), we would obtain the same optimization problem replacing $\alpha^*_t$ by $\mathbb{E}\left[\left. \int_t^{\infty} e^{-\beta (s-t)} dp_s\right|\mathcal{F}^0_t\right]$.} The first term in the definition of $\alpha^*_t $ is the average belief for alpha under the weighted client measure $\lambda_t$, while the second term takes into account mismatches in the time-horizons of the clients. 
For each $t\ge 0$, we define the profitability function $m_t$ by:
\begin{equation}
\label{fo:m}
\alpha \hookrightarrow m_t(\alpha) = (\alpha - \alpha_t^*)[\mu_t([\alpha,\infty)) - 1_{(-\infty,0]}(\alpha)]. 
\end{equation}
For each $\alpha\in\R$, $(m_t(\alpha))_{t\ge 0}$ is a progressively measurable stochastic process and for each fixed $t\ge 0$, the function $\alpha\hookrightarrow m_t(\alpha)$ is almost surely continuous in $\alpha$, except for a possible jump $m_t(0^+)-m_t(0^-)$ at $\alpha=0$. $m_t$ is bounded and vanishes at the infinitives because of the integrability property (\ref{of:integrability}) of the solutions of the SPDE. 
The profitability function quantifies the expected profit for an order placed at time $t$ at the price level $\alpha$. Indeed, the absolute value $|\alpha - \alpha_t^*|$  is equal to the \emph{spread} the market maker expects to gain per filled order, and up to a possible change of sign, the term $ \mu_t([\alpha,\infty))- 1_{(-\infty,0]}(\alpha)$  is equal to the proportion of clients that will fill the order. If the arrivals of the agents of our model occurred  according to a Poisson process instead of simultaneously, this would be the \emph{filling probability} of the order. The respective contributions of these two terms are commonly parsed by practitioners.
Notice also that, in the degenerate case $\epsilon = 0$, the profitability function is the derivative of the Hamiltonian in the direction of the control.

We now identify the modified Hamiltonian in terms of the control $g$ without involving the anti-derivatives $\gamma'_t$ an $\gamma_t$. 

\begin{lemma}
For each $t\ge 0$, we have the identity
\begin{equation}
\label{of:profitability}
 \left\langle  (id- \alpha_t^*)\gamma_t' - \gamma_t, \mu_t\right\rangle = \left\langle m_t, \gamma_t'' \right\rangle .
\end{equation}
\end{lemma}
\begin{proof}
Successive integrations by parts and simplifications yield:
\begin{align*}
\left\langle \gamma_t'', (id - \alpha_t^*)\left(\mu_t([\,\cdot\,,\infty)) - 1_{(-\infty,0]}(\,\cdot\,)\right) \right\rangle 
&= -\left\langle \gamma_t', \left(\mu_t([\,\cdot\, ,\infty)) - 1_{(-\infty,0]}(\,\cdot\,)\right) - (id - \alpha^*) \mu_t  - \alpha_t^* \delta_0 \right\rangle \\
&= \left\langle \gamma_t, \mu_t - \delta_0 \right\rangle + \left\langle \gamma_t' (id-\alpha^*), \mu_t\right\rangle  + \alpha_t^* \gamma'(0) \\
&=  \left\langle (id - \alpha_t^*) \gamma_t' - \gamma_t, \mu_t \right\rangle
\end{align*}
where the first integration by parts is justified by the fact that $\gamma'$ is bounded, the Dirac distribution has compact support and the simple facts:
$$
 \lim_{\alpha \rightarrow \infty} \mu_t([\alpha,\infty))= 0,\quad\text{and}\quad
 \lim_{\alpha \rightarrow -\infty} \mu_t([\alpha,\infty)) - 1_{\alpha\leq 0} = - \lim_{\alpha \rightarrow -\infty} \mu_t((-\infty,\alpha]) = 0,
$$ 
and $\alpha \mu(\alpha)$ vanishes at the infinities. The linear growth of $\gamma$ justifies the second integration by parts. The last line uses the initial conditions of $\gamma_t$ and $\gamma_t'$.
\end{proof}

We can now state and prove the main result of this section:
\begin{theorem}
The solution to the market maker's control problem is given by
\begin{equation}
\label{liquidity}
g_t(\alpha) = \frac{e^{ m_t(\alpha)/\epsilon}}{\int e^{ m_t(\alpha)\epsilon}d\mu_t(\alpha)}
\end{equation}
\end{theorem}
The quantity in the right hand side of \label{liquidity} can be seen as the change of measure from the distribution of the clients alphas to the order book of the market maker.
\begin{proof}
Using Lemma \ref{BSDE_MM}  and equation (\ref{of:profitability}) enables us to rewrite the modified Hamiltonian as:
\begin{equation}
\tilde H(g)= \left\langle m_t g, \mu_t \right\rangle - \epsilon \left\langle f \circ g, \mu_t \right\rangle 
\end{equation}
which we need to maximize for each fixed $t\ge 0$, over $g\in L^1(\mu_t)$ with $\left\langle g, \mu_t\right\rangle =1$ and $g\geq 0$. Again, for each fixed $t$ and almost surely, by convexity of $f$, $\tilde H$ is concave in $g$ and bounded from above, so that $\tilde{M} = \sup_{A_t} \tilde H(g)$ exists and is finite. 

Let $(g_n)_{n\ge 1}$ be a maximizing sequence in $A_t$, i.e. a sequence of admissible $g_n$ such that $H(g_n) \rightarrow \tilde{M}$. By Komlos's lemma (see for example \cite{DelbaenSchachermayer}), there exists a subsequence $g_{\phi(n)}$ and an element $g_\infty \in L^1(\mu_t)$ such that $\frac{1}{n} \sum_{i=1}^n g_{\phi(\psi(i))} \rightarrow g_\infty$ $\mu_t$-a.e. as $n \rightarrow\infty$ for any subsequence of the original subsequence. Clearly, $g_\infty \geq 0$ $\mu_t$ almost everywhere. 
To prove $\left\langle g_\infty, \mu_t\right\rangle = 1$ we show uniform integrability of  $\{g_n\}_n$ using De la Vall\'ee-Poussin's theorem (see \cite{Meyer}). Indeed, by definition of the maximizing sequence, $H(g_n) \geq \tilde{M}-1$ for large enough $n$. Hence $\epsilon \left\langle f \circ g_n, \mu_t \right\rangle \leq - \tilde{M} + 1 + ||m_t||_\infty$. $f$ is non-negative, convex, increasing on $[1,\infty)$ and verifies $\lim_{x\rightarrow\infty} \frac{1}{x}f(x) =\infty$. This concludes the proof of the admissibility of $g_\infty$ as a control. Finally, by concavity of the modified Hamiltonian,  the supremum is attained at $g_\infty$. 

We now characterize the maximal element we just constructed. We introduce a Lagrange multiplier $\eta \in \R$ to relax the problem to the set of  $g \in L^1(\mu_t)$ such that $g \geq 0$, so that the Lagrangian reads:
\begin{equation}
\tilde L(g)=\left\langle (m_t - \eta) g - \epsilon f \circ g , \mu_t \right\rangle 
\end{equation}
Classical variational calculus yields the optimality condition
\begin{equation}
g =  (f')^{-1}\left(\epsilon^{-1}\left(m_t - \eta\right)\right),\label{variational}
\end{equation}
for some $\eta$. The existence and uniqueness of a Lagrange multiplier renormalizing $g$ is obvious given the explicit formula for $(f')^{-1}$. We conclude that $g$ is the desired optimum.
\end{proof}

\subsection{Interpretation}

\begin{figure}[htbp]
	\centering
		\includegraphics[width=1.00\textwidth]{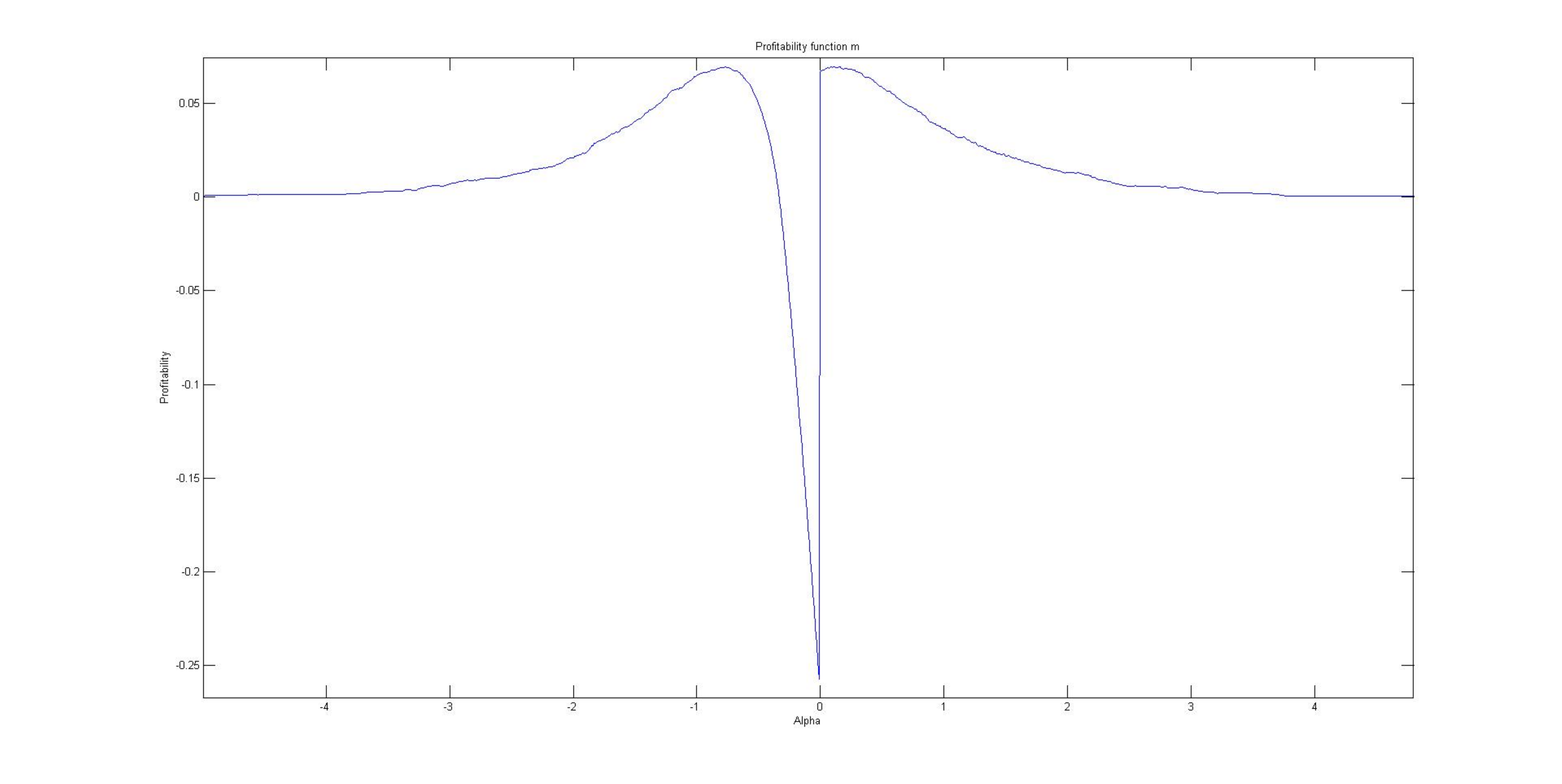}
	\caption{Graph of a simulated profitability curve $m$ with $\alpha^* \approx -0.33$. The maximum profitability is only $5\%$ of the volatility used to simulate prices. The average spread being comparable to the volatility (see \cite{Bouchaud2}), this means that the market maker's margins are very thin in this simulation.}
	\label{fig:profitability}
\end{figure}

\begin{figure}[htbp]
	\centering
		\includegraphics[width=1.00\textwidth]{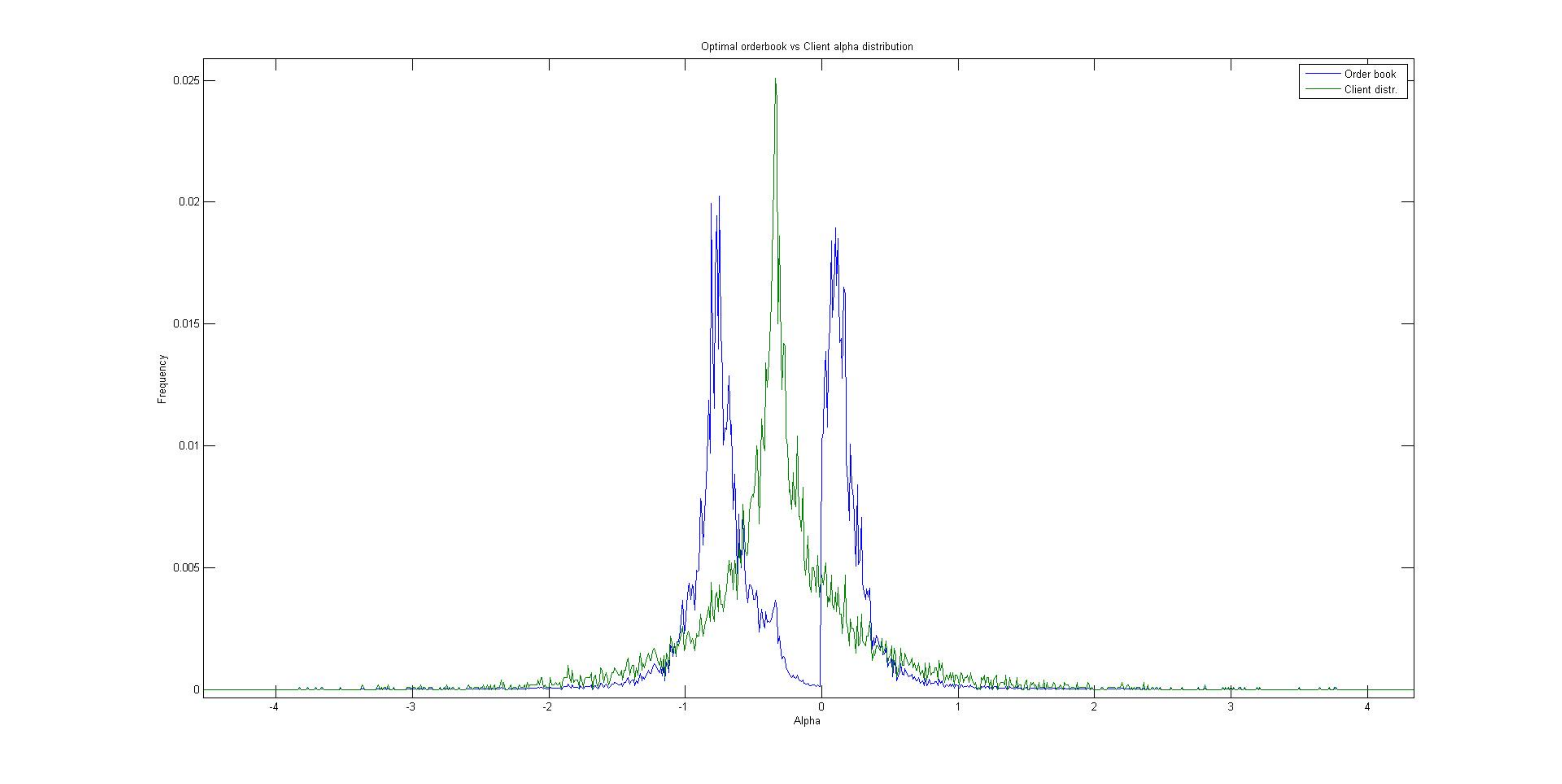}
	\caption{Bimodal distribution: Optimal order book $\gamma''$ for $\epsilon = 0.01$ using the simulated alpha distribution and the entropic penalizing factor. Unimodal distribution: Client alpha distribution. The profitability used is the same as in figure \ref{fig:profitability}.}
	\label{fig:orderbook}
\end{figure}

Figure \ref{fig:profitability} and Figure \ref{fig:orderbook} illustrate what the market maker does:
\begin{enumerate}
\item He tries to stick to a shape not too far away from $\mu_t$, his client alpha distribution. This is to avoid feedback effects and associated errors on the price estimation.
\item He also takes into account the profitability function $m_t$, which leads him ``to make a big hole'' in the center of the distribution $\mu_t$.
\end{enumerate}
The combined effects lead to the familiar ``double hump'' shape of the order book, as seen in \cite{Bouchaud2}. Other consequences of the liquidity formula are
\begin{enumerate}
\item In the limit where one client is perfectly intelligent of the price ($\epsilon \rightarrow 0$), the market maker places a Dirac mass on the order book. In particular, he only trades on profitable sections of the book.
\item In the noise trade limit ($\epsilon \rightarrow \infty$), the market maker simply reproduces the client alpha distribution.
\item If the return distribution of the asset is mean-reverting (which is not the case for options with maturities, for example), then so will $\mu_t$ and hence $\gamma_t$ and $c_t$. In the case of a European option, $\gamma''_t$ converges to a Dirac at the payoff at maturity.
\end{enumerate}

\section{Conclusion}

In this paper, we propose an equilibrium model for high frequency market making. We show that each client chooses his level of trade by an expected discount of future prices according to his beliefs. We call this the implied alpha of the client. This relationship can be used as a codebook to link trade and price dynamics. The different clients can be differentiated according to the time scales of their implied alphas and the variances of the errors they make estimating the price. This leads to an expression for the latter that the market maker can use to avoid having his own view on the market. He can then construct a profitability curve which dictates which sections of the order book are the most profitable. In solving the market maker optimization problem, we use a penalizing term to smooth the objective function and capture possible feedback effects. All these result in a tractable framework in which we can solve the market maker's control problem and identify the equilibrium order book dynamics.

\appendix

\section{A convenient form of the  stochastic maximum principle}
We present a form of the Pontryagin stochastic maximum principle tailored to the needs of the analysis of Section \ref{se:mm}.
The set-up is quite general, following loosely \cite{Pham} Chapter 3. We assume that $(\Omega,\F, \left(\F_t\right)_{t \geq 0}, \P)$ is a filtered probability space, $\left(\F_t\right)_{t \geq 0}$ is generated by a Wiener process $(W_t)_{t \geq 0}$, and we let $A_t(\omega)$ be such that for all $(t, \omega)$, $A_t$ is a Borel convex subset of a Polish topological vector space $E$, adapted to $\left(\F_t\right)_{t \geq 0}$. The admissible control processes $(\alpha_t)_{t\ge 0}$ are the progressively measurable processes in $E$ such that  $\alpha_t\in A_t$ for all $t\ge 0$. We also assume that the dynamics of the controlled state $X$ are governed by an Ordinary Differential Equation(ODE) with random coefficients valued in $\mathbb{R}^n$:
\begin{equation}
dX_t = f_t(X_t, \alpha_t) dt
\end{equation}
where the coefficient $f: (\omega, t, X, a) \rightarrow f_t(X, a)$ is Lipschitz in $X$ uniformly in all the other variables. Assume an estimate of the form
\begin{equation}
\left|X_t\right| \leq |X_0| e^{C t} \label{growth}
\end{equation}
where $C$ is a constant. We also assume that $f$ is continuously differentiable (i.e. $C^1$) in $X$.
The controller's objective function is given by
\begin{equation}
J(\alpha) =\E\left[\int_0^\tau j_t(X_t, \alpha_t) dt + g_\tau( X_\tau)\right]
\end{equation}
where $\tau$ is a stopping time adapted to $\left(\F_t\right)_{t \geq 0}$ and $j$ and $g$ are random real-valued concave functions which are $C^1$ in $X$. Furthermore, we assume that $j$, $g$ and $\tau$ are such that $E\left[|g_\tau( X_\tau)|\right]$ and $\E\left[\int_0^\tau |j_t(X_t, \alpha_t)| dt\right]$ are finite and uniformly bounded from above.
Next, we define the Hamiltonian
\begin{equation}
\mathcal{H}_t(\omega, X, a, Y) = f_t(\omega, X, a) Y + j_t(\omega, X, a)
\end{equation}
for $Y\in\R^n$ and  for each admissible control $(\alpha_t)_{t\ge 0}$, the adjoint equation:
\begin{equation}
-dY_t = \partial_X \mathcal{H}_t(X_t, \alpha_t, Y_t) dt - Z_t dW_t\label{adjoint}
\end{equation}
with $Y_\tau = \partial_X g(X_\tau)$ and $Z_t$ such that $\E\int_0^\tau |Z_t|^2 dt < \infty$. Under these conditions we have the following result.

\begin{theorem}[Pontryagin's maximum principle]
Let $\hat{\alpha}$ be an admissible control and $\hat{X}$ be the associated state variable. Suppose there exists a solution $(Y,Z)$ to the adjoint equation (\ref{adjoint}) such that 
\begin{equation}
\E\left[ \int_0^\tau e^{C t} |Z_t|^2 dt\right] < \infty, \label{Z_growth}
\end{equation}
and the Hamitonian verifies for every$t\ge 0$
\begin{equation}
\mathcal{H}_t\left(\hat{X}_t,\hat{\alpha}_t,Y_t\right) = max_{a \in A_t} \mathcal{H}_t\left(\hat{X}_t,a,Y_t\right),\qquad a.s.
\end{equation}
and for every $t\ge 0$, almost surely, the function
$$
(X,\alpha)\hookrightarrow \mathcal{H}_t\left(\hat{X}_t,\hat{\alpha}_t,Y_t\right) 
$$
is concave, then $\hat{\alpha}$ is an optimal control, that is, $J(\alpha)\leq J(\hat{\alpha})$ for all admissible $\alpha=(\alpha_t)_{t\ge 0}$.
\end{theorem}
\begin{proof}
If $(X_t)_{t\ge 0}$ is the state associated to another admissible control $(\alpha_t)_{t\ge 0}$, we have the chain of relationships:
\begin{eqnarray}
&&\E\left[g_\tau(X_\tau) - g_\tau(\hat{X}_\tau)\right]\notag \\
&&\phantom{?}\leq \E\left[\left(X_\tau - \hat{X}_\tau\right) \cdot Y_\tau \right] \label{P1} \\ 
 &&\phantom{?}= \E\left[ \int_0^\tau \left( - \left(X_t - \hat{X}_t\right) \cdot \partial_X \mathcal{H}_t(\hat{X}_t, \hat{\alpha}_t, Y_t) + \left(f_t(X_t, \alpha_t) - f_t(\hat{X}_t, \hat{\alpha}_t)\right)\cdot Y_t \right)dt\right] \label{P2} \\ 
&&\phantom{?}\leq \E\left[\int_0^\tau \left(\mathcal{H}_t(\hat{X}_t,\hat{\alpha}_t,Y_t) - \mathcal{H}_t(X_t,\hat{\alpha}_t,Y_t) + \left(f_t(X_t, \alpha_t)- f_t(\hat{X}_t, \hat{\alpha}_t)\right)\cdot Y_t \right)dt  \right] \label{P3} \\ 
&&\phantom{?}= \E\left[\int_0^\tau \left(\mathcal{H}_t(\hat{X}_t,\hat{\alpha}_t,Y_t) - \mathcal{H}_t(X_t,\hat{\alpha}_t,Y_t) - \mathcal{H}_t(\hat{X}_t,\hat{\alpha}_t,Y_t) + \mathcal{H}_t(X_t,\alpha_t,Y_t)\right)dt\right] \notag\\ 
&&\phantom{?????????}- \E\left[\int_0^\tau \left( j_t(X_t,\alpha_t)- j_t(\hat{X}_t, \hat{\alpha}_t)\right)dt  \right] \label{P4} \\ 
&&\phantom{?}\leq - \E\left[\int_0^\tau j_t(X_t,\alpha_t) - j_t(\hat{X}_t, \hat{\alpha}_t)\right]. \label{P5} 
\end{eqnarray}
(\ref{P1}) stems from the concavity of $g$ and terminal condition of $Y$. (\ref{P2}) follows from the dynamics of $Y$ and the relationship:
\begin{equation}
\E\left[\int_0^\tau \left(X_t-\hat{X}_t\right)^2 Z_t^2 dt\right] \leq 2 X_0 \E\left[ \int_0^\tau e^{C t} Z_t^2 dt\right] < \infty
\end{equation}
which guarantees that the local martingale part is a martingale. (\ref{P3}) holds by the concavity of the Hamiltonian. (\ref{P4}) is just the definition of $\mathcal{H}$. Finally, (\ref{P5}) is a consequence of the fact that $\hat{\alpha}_t$ maximizes $\mathcal{H}_t(\hat{X}_t, \cdot, Y_t)$.
\end{proof}

\section{Derivation of the SPDE (\ref{OU_SPDE})}
We give the main steps of the proof of Proposition \ref{pr:spde}.
By independence of the common randomness $W$ and the idiosyncratic randomness $(M^i,\alpha_0^i,\epsilon^i)_{i\ge 1}$, we can freeze the randomness of $W$ and work after conditioning with respect to $\F^W$, the $\sigma$-algebra generated by $W$, without affecting the independence properties of the idiosyncratic random variables. By independence of the $M^i$, the $\alpha^i_0$ and $\epsilon^i$, we have that, conditional on $\F^W$, the random variables
\begin{equation}
\alpha^i_t = \alpha^i_0 e^{-\rho t} + \int_0^t e^{-\rho (t-s)} \left(\nu dW_s + \sigma dM^i_s\right)
\end{equation}
are iid and Gaussian. So if $f\in C^2$ is such that $f$ and its two derivatives have at most polynomial growth, given the assumption (\ref{of:m}) on the joint distribution $m$ of all the couples $(\alpha^i_0,\epsilon^i)$, we can apply the law of large numbers and get:
\begin{align*}
\lim_{n\rightarrow\infty} &\frac{1}{n} \sum_{i=1}^n \epsilon^i f(\alpha^i_t) = \E\left[\left.\epsilon^1 f(\alpha^1_t)\right|\F^W\right] \\
																																						&=  \int_{\R^3} \epsilon f\left(\alpha_0 e^{-\rho t} + \int_0^t e^{-\rho(t-s)} \nu dW_s + \sigma x \sqrt{\int_0^t e^{-2\rho(t-s)} ds} \right) \frac{1}{\sqrt{2\pi}} e^{-\frac{x^2}{2}}  m(d\alpha_0,d\epsilon) dx
\end{align*}
The other results can be derived directly from this explicit representation.

\bibliographystyle{plain}

\begin{thebibliography}{10}

\bibitem{Alfonsi}
A.~Alfonsi, A.~Fruth, and A.~Schied.
\newblock Optimal execution strategies in limit order books with general shape
  functions.
\newblock {\em Quantitative Finance}, 10(2):143--157, 2010.

\bibitem{Almgren}
R.~Almgren and N.~Chriss.
\newblock Optimal execution of portfolio transactions.
\newblock {\em Journal of Risk}, 3(2):5--39, 2000.

\bibitem{Almgren2}
R.~F. Almgren.
\newblock Optimal execution with nonlinear impact functions and
  trading-enhanced risk.
\newblock {\em Applied Mathematical Finance}, 10(1):1--18, 2003.

\bibitem{Amihud}
Y.~Amihud and H.~Mendelson.
\newblock Asset pricing and the bid-ask spread.
\newblock {\em Journal of Financial Economics}, 17(2):223--249, 1986.

\bibitem{Bouchaud}
J.~. Bouchaud, Y.~Gefen, M.~Potters, and M.~Wyart.
\newblock Fluctuations and response in financial markets: The subtle nature of
  'random' price changes.
\newblock {\em Quantitative Finance}, 4(2):176--190, 2004.

\bibitem{Bouchaud3}
J.~. Bouchaud, M.~Mézard, and M.~Potters.
\newblock Statistical properties of stock order books: Empirical results and
  models.
\newblock {\em Quant.Finance}, 2(4):251--256, 2002.

\bibitem{Cont}
R.~Cont, S.~Stoikov, and R.~Talreja.
\newblock A stochastic model for order book dynamics.
\newblock {\em Operations research}, 58(3):549--563, 2010.

\bibitem{Csiszar}
I.~Csiszar.
\newblock {\em Information-type measures of difference of probability
  distributions and indirect observations}, volume~2.
\newblock 1967.

\bibitem{DelbaenSchachermayer}
F.~Delbaen and W.~Schachermayer.
\newblock {\em The Mathematics of Arbitrage}.
\newblock Springer Verlag, 2010.

\bibitem{Garman}
M.~B. Garman.
\newblock Market microstructure.
\newblock {\em Journal of Financial Economics}, 3(3):257--275, 1976.

\bibitem{Hasbrouck}
J.~Hasbrouck.
\newblock {\em Empirical market microstructure}.
\newblock Oxford University Press, 2007.

\bibitem{Kurtz}
T.~G. Kurtz and J.~Xiong.
\newblock Particle representations for a class of nonlinear spdes.
\newblock {\em Stochastic Processes and their Applications}, 83(1):103--126,
  1999.

\bibitem{Kurtz2}
T.~G. Kurtz and J.~Xiong.
\newblock Numerical solutions for a class of {SPDE}s with applications to
  filtering.
\newblock {\em Stochastics in Finite and Infinite Dimension}, pages 233--258,
  2000.

\bibitem{Kyle}
A.~S. Kyle.
\newblock Continuous auctions and insider trading.
\newblock {\em Econometrica}, 53(6):1315--1335, 1985.

\bibitem{Meyer}
P.A. Meyer.
\newblock {\em Probabilities and Potential}.
\newblock North-Holland Pub., 1966.

\bibitem{Wang}
A.~Obizhaeva and J.~Wang.
\newblock Optimal trading strategy and supply/demand dynamics.
\newblock {\em Preprint}, 2005.

\bibitem{Pham}
H.~Pham.
\newblock {\em Continuous-time stochastic control and optimization with
  financial applications}.
\newblock Springer Verlag, 2008.

\bibitem{Bouchaud2}
M.~Wyart, J.~. Bouchaud, J.~Kockelkoren, M.~Potters, and M.~Vettorazzo.
\newblock Relation between bid-ask spread, impact and volatility in
  order-driven markets.
\newblock {\em Quantitative Finance}, 8(1):41--57, 2008.

\bibitem{Farmer}
I.~Zovko and J.~D. Farmer.
\newblock The power of patience: A behavioral regularity in limit order
  placement.
\newblock {\em Quantitative Finance}, 2(5):387--392, 2002.

\end{thebibliography}

\end{document}